\documentclass{llncs}
\pdfoutput=1

\usepackage{amsmath,amssymb}

\usepackage{ulem}
\usepackage{subfig}

\usepackage{color}
\usepackage{graphicx}
\usepackage{stmaryrd}
\usepackage{xspace}
\usepackage{clrscode}

\usepackage{tikz}
\usepgflibrary{arrows}
\usetikzlibrary{calc}
\usetikzlibrary{positioning}

\usepackage{wasysym}

\usepackage{mikomath,mikoaux}

\def\theTitle{How to Complete an Interactive Configuration Process?}
\def\theTitleBroken{How to Complete an\\Interactive Configuration Process?}

\def\mj{Mikol\'a\v{s} Janota}
\def\rg{Radu Grigore}
\def\jpms{Joao Marques-Silva}
\def\gb{Goetz Botterweck}

 \title{\theTitleBroken}
 \subtitle{Configuring as Shopping}

\author{\mj\inst{1} \and \gb\inst{2} \and \rg\inst{3} \and \jpms\inst{3}}
\institute{
Lero, University College Dublin, Ireland\and
Lero, University of Limerick, Ireland\and
University College Dublin, Ireland}

\definecolor{citeblue}{rgb}{0.4,0,.6}
\definecolor{refcolor}{rgb}{0,0.1,0.5}
\usepackage[pdftex,a4paper=true,colorlinks=true,bookmarks=false,
            linkcolor=refcolor,citecolor=citeblue,urlcolor=blue]{hyperref}

\hypersetup{
pdftitle={{\theTitle}},
pdfauthor={{\mj}, {\gb}, {\rg}, and, {\jpms}}
}

\newcommand{\refgen}[2] {\hyperref[{#2}]{#1~\ref*{#2}}}

\newtheorem{observation}{Observation}

\definecolor{gray}{rgb}{.4,.4,.4}
\newcommand{\notegengen}[3]{\ {#1}\raisebox{3pt}{\textsf{\tiny #2}}\marginpar{\scriptsize\textsf{#2:} #3}}
\newcommand{\notegen}[2]{\notegengen{\eighthnote}{#1}{#2}}
\renewcommand{\notegengen}[3]{}

\newcommand{\notem}[1]{\notegen{miko}{#1}}

\newcommand{\notegb}[1]{\notegen{gb}{#1}}

\newcommand{\todo}[1]{\notegengen{$\dagger$}{\textcolor{red}{TODO}}{#1}}

\def\V{\ensuremath{{\mathcal{V}}}}

\def\D{\ensuremath{{\mathcal{D}}}}

\DeclareMathOperator{\SAT}{\textsc{Sat}}

\DeclareMathOperator{\minm}{min}

\begin{document}
\normalem
\maketitle


\begin{abstract}
  When configuring customizable software, it is useful to
  provide interactive tool-support that ensures that the configuration
  does not breach given constraints.
  But, when is a configuration complete and how can the tool help the
  user to complete it?  \notegb{I do not remember seeing a discussion
    of the ``When'' part.}
  We formalize this problem and relate it to concepts from
  non-monotonic reasoning well researched in Artificial Intelligence.
  The results are interesting for both practitioners and
  theoreticians.  Practitioners will find a technique facilitating an
  interactive configuration process and experiments
  supporting feasibility of the approach. Theoreticians will find
  links between well-known formal concepts and a concrete practical
  application.
\end{abstract}
\section{Introduction}
  \label{sec:Introduction}

  \emph{Software Product Lines}~(SPLs) build on the assumption that
  when developing software-intensive systems it is advantageous to
  decide upfront which products to include in scope and then manage
  construction and reuse systematically~\cite{ClementsNorthrop02}.

  This approach is suitable for families of products that share a
  significant amount of their user-visible or internal functionality.
  Parnas identified such \emph{program families} as:
\emph{    $\dots$ sets of programs whose common properties are so extensive
    that it is advantageous to study the common properties of the
    programs before analyzing individual members.}~\cite{Parnas76}

  A key aspect of SPLs is that the scope of products is defined and
  described \emph{explicitly} using models of various
  expressivity~\cite{KangEtAl90,Batory05,SchobbensEtAl06,JK2007Reasoning}.
  Conceptually, we can consider an SPL as a mapping from a
  \emph{problem space} to a \emph{solution space}. The problem space
  comprises requirements that members of the product line satisfy, and,
  the solution space comprises possible realizations, \eg programs in
  C$^\text{\tiny{++}}$.  These spaces are defined by some constraints,
  \ie requirements or solutions violating the constraints are not
  within the respective space.

  A specification for a new product is constructed from the
  requirements which must be matched to the constraints that define the scope
  of the product line. In effect, the purchaser picks a particular
  member of the problem space.  Subsequently, software engineers are
  responsible for delivering a product, a member of the solution
  space, corresponding to the given specification.

  If the problem space is complex, picking one of its members is not
  trivial. Hence, we strive to support interactive configuration
  with \emph{configurator} tools.
  Despite the fact that this has been researched extensively
  (see~\cite{LottazStalkerSmith98,HadzicEtAl04,Batory05,derMeerEtAl06,JKW2008Model,Janota2008Do}),
  little attention has been paid to the completion of a configuration
  process. Namely, how shall we treat variables of the configuration
  model that have not been bound by the user at all? (This issue has
  been noted by Batory in~\cite[Section~4.2]{Batory05}).

  This article studies this problem (\autoref{sec:ShoppingPrinciple}) and designs an enhancement of a
  configurator that helps the user to get closer to finishing the
  configuration process by binding variability without making
  decisions for the user, \ie it is aiming at not being overly smart.
  The article focuses mainly on this functionality for propositional
  configuration (\autoref{sec:PropositionalConfiguration}) and it
  relates to research on Closed World Assumption
  (\autoref{sec:DispensableToMinimals}).  The general,
  non-propositional, case is conceptualized relying on the notion of
  preference (\autoref{sec:CSPs}).


\section{Background and Motivation}
  \label{sec:SPL}

  Kang~et~al.\ developed a methodology \emph{Feature Oriented Domain
    Analysis} (\emph{FODA}), where \emph{feature models} are used to
  carry out \emph{domain analysis}---a systematic identification of
  variable and common parts of a family of systems~\cite{KangEtAl90}.
  For the purpose of this article, it is sufficient to consider a
  feature as ``\textit{a prominent or distinctive user-visible aspect,
    quality, or characteristic of a software system or
    system}''~\cite{KangEtAl90}, and a product as a combination of its
  features.  A \emph{Software Product Line} is a system for developing
  products from a certain family captured as a set of feature
  combinations defined by a \emph{feature model}.
\begin{figure}[t]
      \caption{FODA notation and Configuration Processes}
      \centering
      \subfloat[\scriptsize Modeling primitives]{\label{fig:fmLegend}    \begin{tikzpicture}
      [every node/.style={inner sep=2pt,node distance=.3,draw,font=\scriptsize\sf}]
      \node[very thick](root) at (0,3.5) { root feature  } ;

      \node[draw=none, font=\scriptsize\it](xorSign) at (0, 2.8) { xor-group } ; 
      \node(xor) at (0,2.5) {  parent feature } ;

     \node(A) at (-.5,2) { } ;
      \node(B) at (0.5,2) { } ;
      \draw (xor) -- (A) ;      
      \draw (xor) -- (B) ;
      \coordinate (m1) at ($ (xor) !.4! (A)  $);
      \coordinate (m2) at ($ (xor) !.4! (B)  $);
      \draw (m1) .. controls +(-45:2mm) and +(225:2mm) .. (m2) ;

      \node(opt) at (1.8,2) { optional child } ;
      \node(optD) [above=of opt] {};
      \node[text width=1.5cm,text centered](mand) at (1.8,3) { mandatory child } ;
      \node(mandD) [above=of mand] {};

      \draw[o-] (opt) --  (optD);
      \draw[*-] (mand) -- (mandD);
    \end{tikzpicture}}
      \subfloat[\scriptsize Feature model example]{\label{fig:exampleABCDfeatureModel}    \begin{tikzpicture}[xscale=1.1,
                         every node/.style={minimum height=.5cm, minimum width=.3cm,inner sep=0pt,draw,font=\scriptsize} ]

      \node[very thick](X) at (2.5,2.5) { $x$ } ;
      \node(Y) at (1.5,2) { $y$ } ;
      \node(A) at (1,1) { $a$ } ;
      \node(B) at (2,1) { $b$ } ;
      \node(C) at (3,1) { $c$ } ;
      \node(D) at (4,1) { $d$ } ;

      \tikzstyle{every node}=[inner sep=0pt]

      \draw[-*] (X) -- node[inner sep=1pt, above,sloped] {$\scriptstyle\iff$} (Y.north) ;
      \draw[-o] (X) -- node[inner sep=1pt, below,sloped] {$\scriptstyle\Leftarrow$} (C.north) ;
      \draw[-o] (X) -- node[inner sep=1pt, above,sloped] {$\scriptstyle\Leftarrow$} (D.north) ;

      \draw (Y) -- (A) ;      
      \draw (Y) -- (B) ;
      \coordinate (m1) at ($ (Y) !.4! (A)  $);
      \coordinate (m2) at ($ (Y) !.4! (B)  $);
      \draw (m1) .. controls +(-45:2mm) and +(225:2mm) .. (m2) ;
      \node[font=\tiny\it] at (1.5, 1.4) { xor };
    \end{tikzpicture}}
      \subfloat[\scriptsize Its semantics] {\label{fig:exampleABCDfeatureModelSemantics}
        $\begin{array}[b]{l}
          v_x  \land {} \\
          v_y\iff v_x \land   {} \\
          v_c\impl v_x \land v_d\impl v_x \land  {} \\
          v_a\impl v_y \land v_b\impl v_y \land  {} \\
          v_y\impl (v_a\lor v_b) \land  {} \\
          v_y\impl \lnot(v_a\land v_b)  \phantom{{}\land{}}
        \end{array}$
      }%
\\
    \subfloat[\scriptsize Example configuration of \autoref{fig:exampleABCDfeatureModel}]{
          \begin{tikzpicture}[rounded corners=3pt,xscale=1.5,yscale=.9,>=angle 60]

      \node[draw, very thick,inner sep=2pt](last) at (0.3,.4) {\scriptsize $\{x,y,a,c\}$};

       \draw (-.9,-.28) rectangle +(3.5,2.4);
       \draw (-.866,-.25) rectangle +(3.3,1.5);
       \draw (-.813,-.15) rectangle +(3.0,.95);
       \draw[->](-.85, 2.12) -- node[right]{\scriptsize $1$} (-.8, 1.25); 
       \draw[->](-.75, 1.25) -- node[right]{\scriptsize $2$} (-.7, .8); 
       \draw[->](-.65, .8) --  (-.4, .5) node[left]{\scriptsize $3$} -- (last.west); 


      [every node/.style={font=\scriptsize}]
      \node at (1.6,1.9) {\scriptsize $\{x,y,b\}$};
      \node at (0,1.9) {\scriptsize $\{x,y,b,d\}$ }; 
      \node at (1.6,1.5) {\scriptsize $\{x,y,b,c,d\}$};
      \node at (0,1.5) {\scriptsize $\{x,y,b,c\}$};
      \node at (0,1) {\scriptsize $\{x,y,a\}$};
      \node at (1.6,1) {\scriptsize $\{x,y,a,d\}$};
      \node at (1.6,.4) {\scriptsize $\{x,y,a,c,d\}$};

    \end{tikzpicture}
      \label{fig:exampleProcess}
    }
    \subfloat[\scriptsize Configuration schematically]{
          \centering
    \begin{tikzpicture}[yscale=1.1,xscale=.8]
        \node at (3,1) [circle,draw=gray,fill=gray,scale=.1] {};
        \node at (3.3,1) [circle,draw=gray,fill=gray,scale=.1] {};
        \node at (3.6,1) [circle,draw=gray,fill=gray,scale=.1] {};

      \draw (3,1) ellipse (3cm and 1cm);
        \node at (.3cm,1cm) {$\phi_0$};
        \node at (1,1.5) [circle,draw=black,fill=black,scale=.4] {};
        \node at (.6,1.1) [circle,draw=black,fill=black,scale=.4] {};
        \node at (1,.4) [circle,draw=black,fill=black,scale=.4] {};

      \draw (3.6,1) ellipse (2cm and .6cm);
        \node at (1.9,1) {$\phi_2$};

      \node at (1.5,.8) [circle,draw=black,fill=black,scale=.4] {};
      \node at (1.6,1.3) [circle,draw=black,fill=black,scale=.4] {};

      \draw (3.3,1) ellipse (2.5cm and .8cm);
        \node at (1.1cm,1cm) {$\phi_1$};
        \node at (2.4,1.2) [circle,draw=black,fill=black,scale=.4] {};
        \node at (2.5,0.8) [circle,draw=black,fill=black,scale=.4] {};

      \draw[very thick] (4.3,1) ellipse (.5cm and .4cm);
        \node at (4.1,1) [] {$\phi_k$};
        \node at (4.5,1) [circle,draw=black,fill=black,scale=.4] {};
    \end{tikzpicture}
      \label{fig:ConfigurationProcess}
    }
\end{figure}
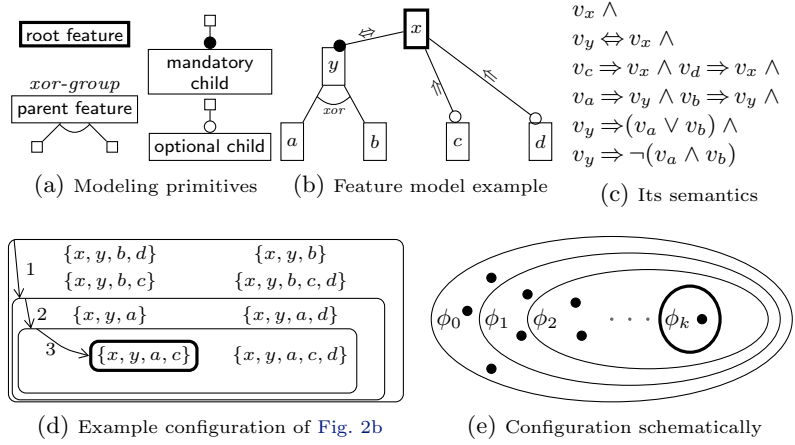
%
  The semantics of a feature model is typically defined with
  propositional logic. Each feature $f$ is represented by a
  propositional variable $v_f$ and a Boolean formula is constructed as
  a conjunction of \flee representing the different modeling
  primitives. The satisfying assignments of the resulting formula
  define the set of possible feature
  combinations~\cite{SchobbensEtAl06}.
  A popular notation is the \emph{FODA notation}~\cite{KangEtAl90} with the
  primitives in \autoref{fig:fmLegend} exemplified by
  \autoref{fig:exampleABCDfeatureModel} whose semantics is
  in~\autoref{fig:exampleABCDfeatureModelSemantics}. The corresponding
  feature combinations, defining the problem space, are listed in
  \autoref{fig:exampleProcess}.
  The FODA notation has several extensions, \eg \emph{feature
    attributes} represent values such as price. For general
  attributes, Boolean logic is insufficient and the semantics is
  expressed as a Constraint Satisfaction
  Problem~\cite{BenavidesEtAl05}.


\subsection{Configuration Process}
  \label{sec:ConfigurationProcess}

  \todo{how do we call the element of the problem space, \ie the
    solution of the constraint in this informal part?}

  In the interactive configuration process, the user specifies his
  requirements step-by-step, gradually shrinking the problem space
  (see \autoref{fig:ConfigurationProcess}).  Hence, this process can
  be seen as a step-wise refinement of the constraint defining the
  space~\cite{CzarneckiEtAl04,JKW2008Model}. How the problem space is
  refined in each step is called a \emph{decision}.

  A \emph{configurator} is a tool that displays the model capturing
  the constraints, presents the user with possible (manual) decisions,
  and infers necessary (automatic) decisions.
  The tool discourages the user from making decisions inconsistent
  with the constraints, and, it suggests decisions that are necessary
  to satisfy the constraints.
  For convenience, configurators enable the user to \emph{retract}
  previously made decisions; some even enable to temporarily
  violate the constraints but this is out of the scope of this article
  and gradual refinement will be assumed.


  For illustration consider the feature model in
  \autoref{fig:exampleABCDfeatureModel} and the following steps (see
  \autoref{fig:exampleProcess}). (1)~The user selects the feature $a$.
  This implicitly deselects $b$ as $a$ and
  $b$ are in an xor-group and there are no feature configurations with both $a$ and $b$.
  (2)~The user selects the feature $c$,
  which has no further effects. (3)~Finally, he deselects the feature
  $d$ and the process is completed since exactly one feature combination is left, \ie each feature
  is either in the product or not.

  In summary, the input to the configurator is a constraint defining a
  large set of possibilities---the outermost ellipse in
  \autoref{fig:ConfigurationProcess}. Gradually, this set is shrunk
  until exactly one possibility is left (assuming the user succeeded).

  \subsection{Completing a Configuration Process and the Shopping Principle}
  \label{sec:ShoppingPrinciple}

  The concepts introduced in the previous sections are
  well known~\cite{HadzicEtAl04}. However, the configuration literature
  does not study the completion of a configuration process.
  As stated above, at the end of the process the decisions must
  determine exactly one feature combination (the innermost ellipse in
  \autoref{fig:ConfigurationProcess}).  So, how does the user achieve
  this?  This article introduces the following classification.

\def\scFull{\textbf{M}\xspace}
\def\scRandom{\textbf{A}\xspace}
\def\scPref{\textbf{A$^+$}\xspace}
\vspace{-.3cm}
  \paragraph{\scFull (Manual).}
  The user makes decisions up to the point when all considered
  variables have been bound, \ie each variable has been assigned a
  value by the user or by a decision inferred by the configurator.
  The disadvantage of this approach is that the user needs to fill in
  every single detail, which is cumbersome especially if there are
  some parts of the problem that are not of a high relevance to the
  user. The only assistance the tool provides is the mechanism that
  infers new decisions or disables some decision. We will not discuss
  this case further.

\vspace{-.3cm}
  \notegb{double check use of vspace here}
  \paragraph{\scRandom (Full blind automation).}
  A function automatically computes \emph{some} values for all
  variables that do not have a value yet. The disadvantage of this
  option is that it takes all the control from the user as it is
  essentially making decisions for him.
\vspace{-.3cm}
  \notegb{double check use of vspace here}
  \paragraph{\scPref (Smart automation).}
  If we assume some form of an a priori, common sense assumptions, there is an approach
  somewhere between the options~\scFull and~\scRandom.
  \notegb{It might not become clear, what you mean by ``knowledge''. How does this affect the automation?}
  In the example above, the user had to \emph{explicitly} deselect the
  optional feature~$d$ but would it be possible to instead say that
  all features not selected should be deselected?
  \todo{make sure that it's in sync with the example}
      The motivation for this approach can be explained by an analogy
      with shopping for groceries (thus the \emph{shopping
        principle}). The customer asks only for those items that he
      wants rather than saying for each item in the store whether he
      wants it or not.\notem{Here we could elaborate on the example by
        extending it with a shopping list.}
      If some variables cannot be bound according to this principle,
      due to some dependencies, the tool will highlight them since it
      is possible that the user forgot to make a certain decision, \eg
      we wouldn't want the tool to decide between features in an
      xor-group ($a$ and $b$).
%

      In some sense, the scenario \scPref is a more careful version of
      scenario \scRandom. Both support the user to complete the configuration process.
      Scenario \scRandom binds \textbf{all} variables, whereas \scPref
      only those for which this doesn't mean deciding something for
      the user.
      The following section investigates these scenarios in
      constraints defined as Boolean \flee (recall that FODA
      produces Boolean \flee).


\section{Propositional Configuration}
  \label{sec:PropositionalConfiguration}

  First, let us recall some basic terms from propositional logic.  Let
  $\V$ be some finite set of variables. The propositional \flee
  discussed from now on will be only on these variables.
%
  A \emph{variable assignment}  assigns either
  \textit{true} or \textit{false} to each considered variable.  Often
  it is useful to think of a variable assignment as the set of
  variables that are assigned the value~\textit{true}, \eg the
  variable assignment $x\mapsto\textit{true}, y\mapsto\textit{false}$
  corresponds to the set $\{x\}$ in the set-based notation.

  A~\emph{model} of a formula $\phi$ is such a variable assignment
  under which $\phi$ evaluates to \textit{true}.  We say that the
  formula $\phi$ is \emph{satisfiable}, denoted as $\SAT(\phi)$, if
  and only if $\phi$ has at least one model, \eg the formula $x\lor y$
  is satisfiable whereas the formula $x\land\lnot x$ is not.  We
  write $\phi\models\psi$ to denote that the formula $\psi$ evaluates
  to \emph{true} under all models of $\phi$, \eg it holds that $x\land
  y\models x$.

\subsection{Propositional Configuration Process}
  \label{PropositionalConfigurationProcess}
  \notegb{This intro sounds a bit strange. Can you express that a bit
    better?}  In order to reason about the feature model and the
  user's requirements, the configurator translates them in some form of
  mathematical representation. In this section we assume that the
  model has already been translated into propositional logic (see
  \autoref{fig:exampleABCDfeatureModelSemantics} for illustration).


  \begin{definition}
    A \emph{propositional configuration process} for some finite set of
    variables $\V$ and a satisfiable propositional formula $\phi$ only
    on the variables from $\V$ is a sequence of propositional \flee
    $\phi_0,\dots,\phi_k$ such that $\phi_0\defeq\phi$,
    $\phi_{i+1}\defeq\phi_i\land \xi_i$ for all $i\in 0\dots k-1$,
    and $\phi_k$ is satisfied by one variable assignment. The \flee
    $\xi_i$ are decisions made by the user or decisions inferred by
    the configurator.
    If $\xi_i$ is of the form $v$ for some variable $v\in\V$, then we
    say that the variable $v$ has been \emph{assigned the value}
    \textit{true} in step $i$; if $\xi_i$ is of the form $\lnot v$, we
    say that it has been \emph{assigned the value} \textit{false} in
    step $i$.
    Observe that $\phi_{i+1}\impl\phi_i$ for all $i\in 0\dots k-1$,
    \ie the set of models is shrunk along the process.
  \end{definition}

  \begin{example}
    Let $\phi_0\defeq (\lnot u\lor\lnot v) \land (x\impl
    y)$. The user sets $u$ to \textit{true} ($\phi_1\defeq\phi_0\land u$);
    the configurator sets $v$ to \textit{false} as $u$ and $v$ are mutually
    exclusive ($\phi_2\defeq\phi_1\land \lnot v$). The user sets
    $y$ to \textit{false} ($\phi_3\defeq\phi_2\land \lnot y$);
    the configurator sets $x$ to \textit{false} ($\phi_4\defeq\phi_3\land\lnot x$).
    The process is finished as all variables were assigned a value.
  \end{example}

  The inference mechanism of the configurator typically inspects for
  all variables $v\in\V$ whether $\phi_l\models v$, in which case it
  sets $v$ to \textit{true}, and whether $\phi_l\models\lnot v$, in
  which case it sets $v$ to \textit{false}. If a value has been
  inferred, the user is discouraged by the user interface to change it
  (``graying out'').  This can be computed with the help of a SAT
  solver~\cite{Janota2008Do} or Binary Decision Diagrams~(BDDs)~\cite{HadzicEtAl04}.

\subsection{Completing a Propositional Configuration Process}
  \label{sec:EndingPropositionalConfigurationProcess}

  Let us look at the scenarios for completing a propositional
  configuration process. Earlier, we have identified two types of
  functions that the user may invoke at any step of the process:
  (\scRandom) a function that binds all the remaining variables;
  (\scPref) a function that finds values for only some variables
  according to an a~priori knowledge; we call this function a
  \emph{shopping principle function}.

  The case~\scRandom is straightforward, finding a solution to the
  formula $\phi_i$ in step $i$ is a satisfiability problem which can
  be solved by a call to a SAT solver or by a traversal of a BDD
  corresponding to $\phi_i$ (see~\cite{Janota2008Do,HadzicEtAl04} for
  further references).

  The scenario~\scPref, however, is more intriguing. The a priori
  knowledge that we apply is the shopping principle (see
  \autoref{sec:ShoppingPrinciple}), \ie what has not been selected
  should be \textit{false}.  According to our experience and intuition
  (as well as other researchers~\cite[Section~4.2]{Batory05}), this
  is well in accord with human reasoning: the user has the impression
  that if a variable (a feature in a feature model) has not been
  selected, then it should be deselected once the process is over.

  \notegb{In the work on the configurator we used the term \emph{eliminated}. \emph{Deselected} might be understood as removing an earlier selection. IMHO, in the context of this paper this is not dangerous and you can use deselectable.}

  However, it is not possible to set all unassigned variables to
  \textit{false} in all cases.  For instance, in $u\lor v$ we cannot
  set both $u$ and $v$ to \textit{false}---the user must choose which
  one should be \textit{true}, and we do not want to make such
  decision for him (otherwise we would be in
  scenario~\scRandom). Another way to see the problem is that setting
  $u$ to \textit{false} will \emph{force} the variable $v$ to be
  \textit{true}, and vice-versa.
  If we consider the formula $x\impl(y\lor z)$, however, all the
  variables can be set to \textit{false} at once and no further input
  from the user is necessary.
  \notegb{We might note that the set of dispensable vars changes during progress of the configuration.}

  In summary, the objective is to maximize the set of variables that can
  be set to~\textit{false} without making any decisions for the user,
  \ie variables that can be deselected safely. Upon a request, the
  configurator will set the safely-deselectable variables to~\textit{false} and
  highlight the rest as they need attention from the user.
%
\subsection{Deselecting Safely}
  \label{sec:DeselectingSafely}

  We start with an auxiliary definition that identifies sets of
  variables that can be deselected at once.  This definition enables
  us to specify those variables that can be deselected safely; we call
  such variables dispensable variables (we kindly ask the reader to
  distinguish the terms deselectable and dispensable).
  \begin{definition}[Deselectable]
    A set of variables $X\subseteq\V$ is \emph{deselectable} w.r.t.\
    the formula $\psi$, denoted as $\D(\psi, X)$, iff all
    variables in $X$ can be set to \textit{false} at once.
    \notegb{here you might add ``without violating the satisfiability of $\psi$''}
    Formally defined as
    $\D(\psi, X) \defeq \SAT\left(\psi \land \bigwedge_{v \in X} \lnot v\right)$.
    Analogously, a single variable $v\in\V$ is \emph{deselectable} w.r.t.\ the formula
    $\psi$ iff it is a member of some deselectable set of
    variables, \ie  $\D(\psi,v)\defeq\SAT(\psi\land\lnot v)$.
  \end{definition}

  \begin{definition}[Dispensable variables]\label{def:dispensable}
    A variable $v\in \V$ is \emph{dispensable} w.r.t.\ a formula
    $\psi$ iff the following holds:
    $\left(\forall X\subseteq\V\right)\left(\D(\psi, X) \impl \D(\psi \land \lnot v, X)\right)$.
  \end{definition}

  In plain English, a variable $v$ is dispensable iff any deselectable
  set of variables~$X$ remains deselectable after $v$ has been
  deselected (set to \textit{false}).  Intuitively, the deselection of
  $v$ does not force selection of anything else, which follows the
  motivation that we will not be making decisions for the user.
  In light of the shopping principle (see
  \autoref{sec:ShoppingPrinciple}), a customer can skip a grocery item
  only if skipping it does not require him to obtain some other items.
  The following examples illustrate the two definitions above.
  \begin{example}
    Let $\phi \defeq (u\lor v) \land (x\impl y)$. Each of the
    variables is deselectable but only $x$ and $y$ are dispensable.
    The set $\{x,y\}$ is deselectable, while the set $\{u,v\}$ is not.
    The variable $u$ is not dispensable as $\{v\}$ ceases to be
    deselectable when $u$ is set to \textit{false}; analogously
    for~$v$. The variable $x$ is dispensable since after $x$ has been
    deselected, $y$ can still be deselected and the variables $u$, $v$ are
    independent of $x$'s value. Analogously, if $y$ is deselected, $x$
    can be deselected.

    Observe that we treat \textit{true} and \textit{false} asymmetrically, deselecting $y$
    forces $x$ to \textit{false}, which doesn't collide with dispensability;
    deselecting $u$ forces $v$ to \textit{true} and therefore $u$ is \emph{not}
    dispensable.
  \end{example}

%
  \begin{example}
    Let $\phi$ be defined as in the previous example and the user is
    performing configuration on it. The user invokes the shopping
    principle function. As $x$ and $y$ are dispensable, both are
    deselected (set to~\textit{false}). The variables $u$ and $v$ are
    highlighted as they need attention. The user selects $u$, which
    results in the formula $\phi_1\defeq\phi\land u$. The variable
    $v$ becomes dispensable and can be deselected automatically. The
    configuration process is finished as all variables have a value.
  \end{example}


  As we have established the term dispensable variable, we continue by
  studying its properties in order to be able to compute the set of
  dispensable variables and to gain more intuition about them.



\begin{lemma}
  Let $\Upsilon_\phi$ denote the set of all dispensable variables of $\phi$.
  For a satisfiable $\psi$, the dispensable variables can be
  deselected all at once, \ie $\D(\psi, \Upsilon_\phi)$.
\end{lemma}
\begin{proof}
  By induction on the cardinality of subsets of $\Upsilon_\psi$.
  Let $\Upsilon_0\defeq \emptyset$, then $\D(\psi, \Upsilon_0)$ as
  $\psi$ is satisfiable.
  Let $\Upsilon_i,\Upsilon_{i+1} \subseteq\Upsilon_\psi$ s.t.\
  \hbox{$\Upsilon_{i+1}= \Upsilon_i \union \{x\}$},
  $|\Upsilon_i|=i$, and \hbox{$|\Upsilon_{i+1}|=i+1$}.
  Since $x$ is dispensable and $\D(\phi, \Upsilon_i)$, then $\D(\phi \land
  \lnot x, \Upsilon_i)$, which is equivalent to $\D(\phi, \Upsilon_i \union
  \{x\})$.
\end{proof}

The following lemma reinforces that the definition of dispensable
variables adheres to the principles we set out for it, \ie it
maximizes the number of deselected variables while not arbitrarily
deciding between variables.

\begin{lemma}
  The set $\Upsilon_\phi$---the set of all dispensable variables of
  $\phi$---is the intersection of all maximal sets of deselectable
  variables of $\phi$.
\end{lemma}
\begin{proof}[sketch]
  From definition of dispensability, any deselectable set remains
  deselectable after any dispensable variable is added to it, hence
  $\Upsilon_\phi$ is a subset of any maximal deselectable set.
  $\Upsilon_\phi$ is a maximal set with this property because for each
  deselectable set that contains at least one non-dispensable variable
  there is another deselectable set that does not contain this
  variable.
\end{proof}

\subsection{Dispensable Variables and Non-monotonic Reasoning}
  \label{sec:DispensableToMinimals}

  After defining the shopping principle in mathematical terms, the
  authors of this article realized that dispensable variables
  correspond to certain concepts from Artificial Intelligence as shown
  in this subsection.

  The \emph{Closed World Assumption~(CWA)} is a term from logic programming
  and knowledge representation.  Any inference that takes place builds on
  the assumption that if something has not been said to be
  \textit{true} in a knowledge base, then it should be assumed
  \textit{false}.  Such reasoning is called \emph{non-monotonic} as an
  increase in knowledge does not necessarily mean an increase in
  inferred facts.
  In terms of mathematical logic, CWA means adding negations of
  variables that should be assumed \textit{false} in the reasoning
  process. Note that not all negateable variables ($\phi\nvDash v$)
  can be negated, \eg for the formula $x\lor y$ both $x$ and $y$
  are negateable but negating both of them would be inconsistent with the formula.

  The literature offers several definitions of reasoning under Closed
  World Assumption~\cite{CadoliLenzerini90,EiterGottlob93}. A
  definition relevant to this article is the one of the
  \emph{Generalized Closed World Assumption (GCWA)} introduced by
  Minker~\cite{Minker82} (see~\cite[Def.~1]{CadoliLenzerini90}).
  \begin{definition}
    The variable $v$ is \emph{free of negation} in the formula $\phi$
    iff for any positive clause $B$ for which $\phi\nvDash
    B$, it holds that $\phi\nvDash v\lor B$.
    The \emph{closure} $C(\phi)$ of a formula $\phi$ is defined as
      $C(\phi) \defeq \phi \union \comprehen{\lnot K}{K \text{ is free
          for negation in } \phi}$.
  \end{definition}

  It is not difficult to see that dispensable variables are those that
  are free of negation as shown by the following lemma.

  \begin{lemma}
    Dispensable variables coincide with those that are free of
    negation.
  \end{lemma}
  \begin{proof}
    Observe that $\phi\nvDash \psi$ iff $\SAT(\phi \land \lnot\psi)$,
    then the definition above can be rewritten as: For $B'\defeq
    \bigwedge_{v\in V_B} \lnot v$ for some set of variables $V_B$ for
    which $\SAT(\phi\land B')$, it holds that $\SAT(\phi\land\lnot
    v\land B')$. According to the definition of $\D$, this is
    equivalent to $\D(\phi, V_B) \impl \D(\phi\land\lnot v, V_B)$
    (compare to \autoref{def:dispensable}).
  \end{proof}

  \emph{Circumscription}, in our case the \emph{propositional
    circumscription}, is another important form of
  reasoning~\cite{McCarthy80}.
%
  A circumscription of a propositional formula $\phi$ is a set of
  minimal models of $\phi$. Where a model $\alpha$ of a formula $\phi$
  is \emph{minimal} iff $\phi$ has no model $\alpha'$ which
  would be a strict subset of $\alpha$,
  e.g., the formula $x\lor y$ has the models $\{x\},\{y\},\{x,y\}$
  where only $\{x\}$ and $\{y\}$ are minimal.  We write $\phi
  \models_{\minm} \psi$ to denote that $\psi$ holds in all minimal
  models of $\phi$, e.g., $x\lor y\models_{\minm}\lnot(x\land y)$.

  The following lemma relates minimal models to dispensable variables
  (The proof of equivalence between minimal models and GCWA is found
  in~\cite{Minker82}).

\begin{lemma}\label{lem:DispensableToMinimal}
  A variable $v$ is dispensable iff it is \textit{false}
  in all minimal models.
\end{lemma}

\begin{example}
  Let $\phi_0\defeq (u\lor v) \land (x\impl y)$. The minimal models of
  the formula $\phi_0$ are $\{u\}$, $\{v\}$, hence
  $\phi_0\models_{\minm}\lnot x$ and $\phi_0\models_{\minm}\lnot
  y$. Then, if the user invokes the shopping principle function, $x$
  and $y$ are deselected, \ie $\phi_1\defeq \phi_0\land \lnot x\land
  \lnot y$. And, the user is asked to resolve the competition between
  $u\lor v$, he selects $u$, resulting in the formula
  $\phi_2\defeq\phi_1\land u$ with the models $\{u\}$ and $\{u,v\}$
  where only the model $\{u\}$ is minimal hence $v$ is set to
  \textit{false} as dispensable. The configuration process is complete
  because $u$ has the value \textit{true} and the rest are
  dispensable.
\end{example}

  \begin{table}[t]
    \centering
  \caption{Experimental Results}\label{tab:ExperimentalResults}
  \begin{tabular}{|l|c|c|c|c|c|c|}
  \hline
  \textbf{Name} &
  \textbf{Features} &
  \textbf{Clauses} &
  \textbf{Length} &
  \textbf{Done} &
  \textbf{Minimal models} \\\hline
  tightvnc &$21$ &$22$ &$5.5$ &$5.5$ &$1.0\pm0.0$ \\\hline
  apl &$27$ &$41$ &$12.2$ &$11.9$ &$1.0\pm0.0$ \\\hline
  gg4 &$58$ &$139$ &$10.0$ &$3.8$ &$15.3\pm22.6$ \\\hline
  berkeley &$94$ &$183$ &$26.6$ &$17.9$ &$1.7\pm1.1$ \\\hline
  violet & $170$ & $341$ & $56.1$ & $47.1$ & $1.6\pm0.9$ \\\hline
  \end{tabular}
  \end{table}

\subsection{Experimental Results}
  \label{sec:experimentalResults}

  The previous section shows that dispensable variables can be found
  by enumerating minimal models. Since the circumscription problem is
  $\Pi^P_2$-complete~\cite{EiterGottlob93} it is important to check if
  the computation is feasible in practice.  We applied a simple
  evaluation procedure to five feature 
  models\footnote{from
  \url{http://fm.gsdlab.org/index.php?title=Model:SampleFeatureModels}}:
  For each feature model we simulated $1000$ random manual configuration 
  processes (scenario \scFull). At each step we enumerated minimal models.
  (Algorithmic details can be found online~\cite{algos_tr}.)
  We also counted how many times 
  there was exactly one minimal model: At those steps the
  configuration process would have been completed if the user
  invoked the shopping principle function.
  
  The results appear in \autoref{tab:ExperimentalResults}. The
  column \textbf{Length} represents the number of user decisions
  required if the shopping principle function is not invoked;
  the column \textbf{Done} represents in how many steps an
  invocation of the shopping principle function completes the
  configuration; the column \textbf{Minimal models} shows that
  the exponential worst case tends not to occur in practice and
  therefore enumeration of all minimal models is feasible. 

\section{Beyond Boolean Constraints}
  \label{sec:CSPs}

  The previous section investigated how to help a user with
  configuring propositional constraints. Motivated by the shopping
  principle, we were trying to set as many variables to \textit{false}
  as possible. This can be alternatively seen as that the user
  \textbf{prefers} the undecided variables to be \textit{false}.

  This perspective helps us to generalize our approach to the case of
  non-propositional constraints under the assumption that there is
  some notion of preference between the solutions. First, let us
  establish the principles for preference that are assumed for this
  section. (1)~It is a partial order on the set in question.  (2)~It is
  static in the sense that all users of the system agree on it, \eg it
  is better to be healthy and rich than sick and poor.  (3)~If two
  elements are incomparable according to the ordering, the automated
  support shall not decide between them, instead the user shall be
  prompted to resolve it.
  \notegb{here we could explain \textbf{how} the proplogic case is a special case of this.}

  To be able to discuss these concepts precisely, we define them in
  mathematical terms.  We start by a general definition of the problem
  to be configured, \ie the initial input to the configurator,
  corresponding to the set of possibilities that the user can
  potentially reach---the outermost ellipse in
  \autoref{fig:ConfigurationProcess}.
  \begin{definition}[Solution Domain]\label{def:SD}
    A \emph{Solution Domain (SD)} is a triple
    $\triple{\V}{\D}{\phi}$ where $\V$ is a set of variables
    $\V=\{v_1,\dots,v_n\}$, $\D$ is a set of respective domains
    $\D=\{D_1,\dots,D_n\}$, and the constraint
    $\phi\subseteq{D_1\times\dots\times D_n}$ is an n-ary relation on
    the domains (typically defined in terms of variables from~$\V$).

    A \emph{variable assignment} is an n-tuple $\tuple{c_1,\dots,c_n}$
    from the Cartesian product $D_1 \times\dots\times D_n$, where the
    constant $c_i$ determines the value of the variable $v_i$ for
    $i\in 1\dots n$.
    For a constraint $\psi$, a variable assignment $\alpha$ is a
    \emph{solution} iff it satisfies the constraint, \ie
    $\alpha\in\psi$.

    An \emph{Ordered Solution Domain~(OSD)} is a quadruple
    \hbox{$\quadruple{\V}{\D}{\phi}{\prec}$} where
    $\triple{\V}{\D}{\phi}$ is an SD and $\prec$ is a partial order on
    \hbox{$D_1\times\dots\times D_n$}.
    For a constraint $\psi$, a solution $\alpha$ is \emph{optimal} iff
    there is no solution $\alpha'$ of $\psi$ s.t.\ $\alpha'\neq \alpha$
    and $\alpha'\prec \alpha$.
  \end{definition}







  Recall that the user starts with a large set of potential solutions,
  gradually discards the undesired ones until only one solution is
  left. From a formal perspective, solution-discarding is carried out
  by strengthening the considered constraint, most typically by
  assigning a fixed value to some variable.

  \begin{definition}[Configuration Process]\label{def:ConfigurationProcess}
    Given a Solution Domain~$\triple{\V}{\D}{\phi}$, an
    \emph{interactive configuration process} is a sequence of
    constraints $\phi_0,\dots,\phi_k$ such that $\phi_0\defeq
    \phi$ and $\card{\phi_k}=1$.
    The constraint $\phi_{j+1}$ is defined as
    $\phi_{j+1}\defeq\phi_j \intersect \xi_j$ where the constraint
    $\xi_j$ represents the decision in step $j$ for $j\in 0\dots k-1$.
%
%
    If $\xi_j$ is of the form $v_i=c$ for a variable $v_i$ and a
    constant $c\in D_i$, we say that the variable
    $v_i$ has been assigned the value $c$ in step $j$.
    Observe that $\phi_{j+1}\subseteq\phi_j$ for $j\in 0\dots k-1$
    and $\phi_j\subseteq\phi$ for $j\in 0\dots k$.
  \end{definition}

  A configurator in this process disables certain values or assigns
  them automatically. In particular, the configurator disallows
  selecting those values that are not part of any solution of the
  current constraint, \ie in step $l$ it disables all values $c\in
  D_i$ of the variable $v_i$ for which there is no solution of the
  constraint $\phi_l$ of the form $\tuple{c_1,\dots, c,\dots\,
    c_n}$. If all values but one are disabled for the domain $D_i$,
  then the configurator automatically assigns this value to the
  variable $v_i$.

  \todo{add some more flow and punchline to this subsection}

  Now as we have established the concept for general configuration,
  let us assume that a user is configuring an Ordered Solution Domain
  (\autoref{def:SD}) and we wish to help him with configuring
  variables that have lesser importance for him, similarly as we did
  with the shopping principle. The configuration proceeds as normal
  except that after the user configured those values he wanted, he
  invokes a function that tries to automatically configure the unbound
  variables using the given preference.

  The assumption we make here is that the variables that were not
  given a value yet should be configured such that the result
  is optimal while preserving the constraints given by the user so
  far. Since the preference relation is a partial order, there may be
  multiple optimal solutions. As we do not want to make a choice for
  the user, we let him focus only on optimal solutions.

  If non-optimal solutions shall be ruled out, the configurator
  identifies such values that never appear in any optimal solution to
  reduce the number of decisions that the user must focus on. Dually,
  the configurator identifies values that appear in all optimal
  solutions, the following definitions establish these concepts.
\vspace{-.5cm}
  \begin{definition}[Settled variables.]\label{def:settled}
    For a constraint $\psi$ and a variable $v_i$, the value $c\in D_i$
    is \emph{non-optimal} iff the variable $v_i$ has the
    value $c$ only in \emph{non-optimal} solutions of $\psi$ (or,
    $v_i$ has a different value from $c$ in all optimal solutions of
    $\psi$).
    \notegb{consistency with textual explanation, ``ruled out''}
%
%
    A value $c$ is \emph{settled} iff $v_i$ has the value~$c$ in all
    optimal solutions of $\psi$.  A variable $v_i$ is settled if there
    is a settled value of $v_i$.

\notegb{consistency with textual
      explanation, ``appears in all''}
  \end{definition}


  \begin{observation}
    For some constraint and the variable $v_i$, a value $c\in D_i$ is
    settled iff all values $c'\in D_i$ different from $c$
    are non-optimal.
  \end{observation}
  \begin{example}
    Let $x,y,z\in \{0,1\}$. Consider a constraint requiring that at
    least one of $x$,$y$,$z$ is set to $1$ (is selected). The
    preference relation expresses that we prefer lighter and
    cheaper solutions where $x$, $y$, and $z$ contribute to the total
    weight by $1$, $2$, $3$ and to the total price by $10$, $5$, and
    $20$, respectively. Hence, the solutions satisfy $(x+y+z>0)$, and
    $\triple{x_1}{y_1}{z_1} \prec \triple{x_2}{y_2}{z_2}$ iff
    $(10x_1+5y_1+20z_1 \leq
     10x_2+5y_2+20z_2) \land
     (1x_1+2y_1+3z_1 \leq
     1x_2+2y_2+3z_2) $.
     Any solution setting $z$ to $1$ is non-optimal as $z$ is more
     expensive and heavier than both $x$ and $y$, and hence the
     configurator sets $z$ to $0$ (it is settled). Choosing between
     $x$ and $y$, however, needs to be left up to the user because $x$
     is lighter than $y$ but more expensive than $y$.
  \end{example}

  Propositional configuration, studied in the previous section, is a
  special case of a Solution Domain configuration with the variable
  domains $\{ \textit{true}, \textit{false} \}$. The following
  observation relates settled and dispensable variables (definitions
  \ref{def:settled}, \ref{def:dispensable}).
\vspace{-.5cm}
  \begin{observation}
    For a Boolean formula understood as an OSD with the preference
    relation as the subset relation, a variable is settled iff it is
    dispensable or it is \textit{true} in all models (solutions).
    Additionally, if each variable is settled, the invocation of the
    shopping principle function completes the configuration process.
  \end{observation}

  This final observation is an answer to the question in the title,
  \ie configuration may be completed when all variables are
  settled. And, according to our experiments this happens frequently
  in practice (column \textbf{Done} in
  \autoref{tab:ExperimentalResults}).

\section{Related Work}
  \label{sec:RelatedWork}

  Interactive configuration as understood in this article has been
  studied \eg by Had\v{z}i\'{c}~et~al.~\cite{HadzicEtAl04},
  Batory~\cite{Batory05}, and Janota~\cite{Janota2008Do}. In an
  analogous approach Janota~et~al.~\cite{JKW2008Model} discuss the use
  of interactive configuration for feature model construction.
  \notegb{Cite VAMOS2009 paper?}  The work of
  van~der~Meer~et~al.~\cite{derMeerEtAl06} is along the same lines but
  for unbounded spaces. Lottaz~et~al.~\cite{LottazStalkerSmith98}
  focus on configuration of non-discrete domains in civil engineering.

  There is a large body of research on \emph{product configuration}
  (see \cite{SabinWigel98} for an overview), which typically is
  conceptualized rather as a programming paradigm than a
  human-interaction problem. Moreover, the notion rules are used
  instead of \flee. Similarly as do we, Junker~\cite{Junker01} applies
  preference in this context.
  We should note that preference in logic has been studied
  extensively, see~\cite{DelgrandeEtAl04}.

  The problem how to help the user to finish the configuration process
  was studied by  Krebs~et~al.~\cite{KrebsEtAtl03} who applied machine
  learning to identify a certain plan in the decisions of the user.

  Circumscription has been studied extensively since the
  80's~\cite{McCarthy80,Minker82,EiterGottlob93}.
%
%
  Calculation of propositional circumscription was studied by Reiter
  and Kleer~\cite{ReiterdeKleer87}; calculation of all minimal models
  by Kavvadias~et~al. and work referenced
  therein~\cite{KavvadiasEtAl00}.


  \section{Summary}
   \label{sec:Summary}

   This article proposes a novel extension for configurators---the shopping
   principle function
   (\autoref{sec:EndingPropositionalConfigurationProcess}). This
   function automates part of the decision-making but is not trying to
   be \emph{too} smart: it does not make decisions between equally
   plausible options.
   The article mainly focuses on the propositional case, as software
   engineering models' semantics are typically propositional. The
   relation with GCWA, known from Artificial Intelligence, offers ways
   how to compute the shopping principle function
   (\autoref{sec:DispensableToMinimals}). Several experiments were
   carried out suggesting that the use of the shopping principle
   function is feasible and useful
   (\autoref{sec:experimentalResults}).  The general,
   non-propositional, case is studied at a conceptual level opening
   doors to further research (\autoref{sec:CSPs}). The authors are
   planning to integrate this function into a configurator and carry
   out further experiments as future work.

\paragraph{\textbf{Acknowledgment}}
This work is partially supported by Science Foundation Ireland under
grant no. {03/CE2/I303\_1} and the IST-2005-015905 MOBIUS project.
The authors thank Don Batory and Fintan Farmichael for valuable
feedback.

\bibliographystyle{abbrv}
\bibliography{refs,proceedings}
\end{document}